\documentclass{amsart}

\usepackage{graphicx}          
\usepackage{graphicx}          

\include{diagrams}
\usepackage{diagrams}
\usepackage{tikz}
\usetikzlibrary{automata}
\usepackage{subfigure}
\usepackage{amssymb,latexsym,amsfonts,amsmath}
\usepackage{graphicx}
\usepackage[boxed]{algorithm2e}
\newtheorem{theorem}{Theorem}[section]
\newtheorem{lemma}[theorem]{Lemma}
\newtheorem{problem}[theorem]{Problem}
\newtheorem{proposition}[theorem]{Proposition}

\newtheorem{definition}[theorem]{Definition}

\newcommand{\delay}{\mathrm{delay}}
\newcommand{\ctrl}{\mathrm{ctrl}}
\newcommand{\send}{\mathrm{send}}
\newcommand{\scc}{\mathrm{sc}}
\newcommand{\ca}{\mathrm{ca}}
\newcommand{\req}{\mathrm{req}}
\newcommand{\alt}{\mathrm{alt}}

\include{diagrams}
\usepackage{diagrams}
\usepackage{tikz}
\usetikzlibrary{automata}
\usepackage{subfigure}
\usepackage{amssymb,latexsym,amsfonts,amsmath}
\usepackage{graphicx}

\begin{document}


%

\begin{abstract}
The research area of Networked Control Systems (NCS) has been the topic of intensive study in the last decade. In this paper we give a contribution to this research line by addressing symbolic control design of (possibly unstable) nonlinear NCS with specifications expressed in terms of automata. We first derive symbolic models that are shown to approximate the given NCS in the sense of (alternating) approximate simulation. We then address symbolic control design with specifications expressed in terms of automata. We finally derive efficient algorithms for the synthesis of the proposed symbolic controllers that cope with the inherent computational complexity of the problem at hand.
\end{abstract}

\title{Integrated Symbolic Design of \\
Unstable Nonlinear Networked Control Systems}\thanks{$^{\ast}$ The research leading to these results has been partially supported by the Center of Excellence DEWS and received funding from the European Union Seventh Framework Programme [FP7/2007-2013] under grant agreement n. 257462 HYCON2 Network of excellence.}\thanks{$^{\ast}$ E-mail: alessandro.borri@iasi.cnr.it, \{giordano.pola,mariadomenica.dibenedetto\}@univaq.it}
\thanks{$^{1}$ Istituto di Analisi dei Sistemi ed Informatica ``A.
Ruberti'', Consiglio Nazionale delle Ricerche (IASI-CNR), 00185 Rome, Italy}
\thanks{$^{2}$ Department of Information Engineering, Computer Science and Mathematics, Center of Excellence DEWS,
University of L{'}Aquila, 67100 L{'}Aquila, Italy}

%
\author[Alessandro Borri, Giordano Pola, Maria D. Di Benedetto]{
Alessandro Borri$^{1}$, Giordano Pola$^{2}$, Maria D. Di Benedetto$^{2}$}
\maketitle


\section{Introduction}

Networked Control Systems (NCS) are complex, heterogeneous, spatially distributed systems where physical processes interact with distributed computing units through non--ideal communication networks. The complexity and heterogeneity of such systems is given by the interaction of at least three components: a plant process that is often described by continuous dynamics, a controller implementing algorithms on microprocessors for the control of the plant, and a communication network conveying information between the plant and the controller which is often characterized by non-idealities such as variable sampling/transmission intervals, variable communication delays, quantization errors, packet dropouts, communication protocol and limited bandwidth. In the last decade, NCS have been the object of great interest in the research community and important results have been achieved, see e.g. \cite{HeemelsSurvey} and the references therein. Most of the results on NCS mainly deals with stabilization problems under an imperfect communication network comprising a subset of the aforementioned communication non-idealities. The work in \cite{BorriHSCC12} instead, considers all the aforementioned communication non-idealities and proposes control algorithms for solving problems with complex specifications expressed in terms of automata.
The main drawbacks of the results reported in \cite{BorriHSCC12} are:

\begin{itemize}
\item[(i)] The plant in the NCS is supposed to be stable, which is quite restrictive in many application domains of interest.
\item[(ii)] The controllers proposed require a large computational complexity in their design.
\end{itemize}

The present work improves the results established in \cite{BorriHSCC12} in two directions:

\begin{itemize}
\item[(i')] We extend our results to possibly unstable nonlinear networked control systems;
\item[(ii')] We design efficient algorithms that cope with the computational complexity of the approach in \cite{BorriHSCC12}.
\end{itemize}

For (i') we generalize the results reported in \cite{MajidTAC11} from nonlinear control systems to nonlinear \textit{networked} control systems. For (ii') we generalize the control algorithms we proposed in \cite{PolaTAC12} for stable nonlinear control systems to \textit{unstable} nonlinear \textit{networked} control systems. 

\section{Notation}\label{sec:Notation}
The symbols $\mathbb{N}$, $\mathbb{N}_0$, $\mathbb{Z}$, $\mathbb{R}$, $\mathbb{R}^{+}$ and $\mathbb{R}_{0}^{+}$ denote the set of natural, nonnegative integer, integer, real, positive real, and nonnegative real numbers, respectively. Given a set $A$ we denote $A^{2}=A\times A$ and $A^{n+1}=A\times A^{n}$ for any $n\in \mathbb{N}$. Given an interval $[a,b]\subseteq \mathbb{R}$ with $a\leq b$ we denote by $[a;b]$ the set $[a,b]\cap \mathbb{N}$. We denote by $\lceil x \rceil=\min\{{n\in\mathbb{Z} \vert n\geq x}\}$ the ceiling of a real number $x$. Given a vector $x\in\mathbb{R}^{n}$ we denote by $\Vert x\Vert$ the infinity norm and by $\Vert x\Vert_{2}$ the Euclidean norm of $x$.
Given $\mu\in\mathbb{R}^{+}$ and $A\subseteq \mathbb{R}^{n}$, we set $[A]_{\mu}=\mu\mathbb{Z}^{n} \cap A$; if $B=\bigcup_{i\in [1;N]}A^{i}$ then $[B]_{\mu}=\bigcup_{i\in [1;N]} ([A]_{\mu})^{i}$. Consider a bounded set $A \subseteq \mathbb{R}^n$ with interior. Let $H=[a_1,b_1]\times[a_2,b_2]\times \dots \times [a_n,b_n]$ be the smallest hyperrectangle containing $A$ and set $\hat{\mu}_{A}=\min_{i=1,2,\dots,n} (b_i-a_i)$. It is readily seen that for any $\mu \leq \hat{\mu}_A$ and any $a\in A$ there always exists $b\in [A]_{\mu}$ such that $\Vert a-b \Vert \leq \mu$.
Given $a\in A\subseteq \mathbb{R}^{n}$ and a precision $\mu\in\mathbb{R}^{+}$, the symbol $[a]_{\mu}$ denotes a vector in $\mu \, \mathbb{Z}^{n}$ such that $\Vert a-[a]_{\mu} \Vert \leq \mu/2$. Any vector $[a]_{\mu}$ with $a\in A$ can be encoded by a finite binary word of length 
$ \lceil\log_{2} \vert [A]_{\mu} \vert \rceil$.
Given a pair of sets $A$ and $B$ and a relation $\mathcal{R}\subseteq A\times B$, the symbol $\mathcal{R}^{-1}$ denotes the inverse relation of $\mathcal{R}$, i.e. $\mathcal{R}^{-1}=\{(b,a)\in B\times A:( a,b)\in \mathcal{R}\}$. The cardinality of a finite set $A$ is denoted by $|A|$.

\section{Networked Control Systems}\label{sec:modelingNCS}
\begin{figure}
\begin{center}
\includegraphics[scale=1]{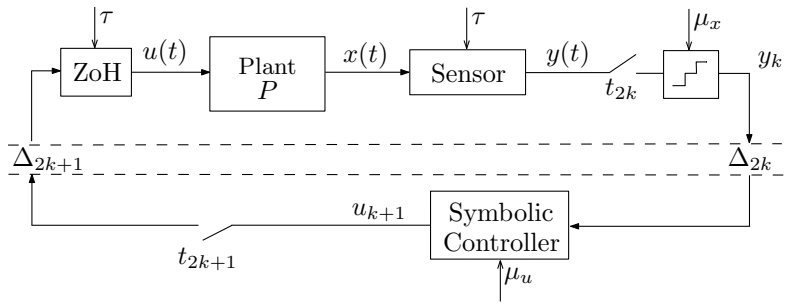}
\caption{Networked control system.}
\label{NCSpic}
\end{center}
\end{figure}
The class of Network Control Systems (NCS) that we consider in this paper has been introduced in \cite{BorriHSCC12}. In this section we briefly review this model. For more details the interested reader is referred to \cite{BorriHSCC12}. The network scheme of the NCS is depicted in Figure \ref{NCSpic}. The direct branch of the network includes the plant $P$, that is a nonlinear control system of the form:
\begin{equation}
\left\{
\begin{array}
{l}
\dot{x}(t)=f(x(t),u(t)),\\
x\in X\subseteq \mathbb{R}^{n},\\
x(0)\in X_{0}\subseteq X,\\
u(\cdot)\in \mathcal{U},
\end{array}
\right.
\label{NCSeq1}
\end{equation}
where $x(t)$ and $u(t)$ are the state and the control input at time $t\in\mathbb{R}^{+}_{0}$, $X$ is the state space, $X_{0}$ is the set of initial states and $\mathcal{U}$ is the set of control inputs that are supposed to be piecewise--constant functions of time from intervals of the form \mbox{$]a,b[\subseteq\mathbb{R}$} to $U\subseteq \mathbb{R}^{m}$. We suppose that sets $X$ and $U$ are convex, bounded and with interior.
 The function $f:X\times U \rightarrow X$ is such that $f(0,0)=0$ and assumed to be Lipschitz on compact sets.
In the sequel we denote by $\mathbf{x}(t,x_{0},u)$ the state reached by (\ref{NCSeq1}) at time $t$ under the control input $u$ from the initial state $x_{0}$; this point is uniquely determined, since the assumptions on $f$ ensure existence and uniqueness of trajectories. We assume that the control system $P$ is forward complete, namely that every trajectory is defined on an interval of the form $]a,\infty\lbrack$.
On the two sides of the plant $P$ in Figure \ref{NCSpic}, a Zero-order-Holder (ZoH) and a (ideal) sensor are placed. We assume that the ZoH and the sensor are synchronized and update their output values at times that are integer multiples of the same interval $\tau\in\mathbb{R}^{+}$, i.e. $u(s \tau+t)=u(s\tau)$, $y(s \tau+t)=y(s\tau)=x(s\tau)$, $t\in [0,\tau[$, $s\in \mathbb{N}_0$, where $s$ is the index of the sampling interval (starting from $0$).
The evolution of the NCS is described iteratively in the following, starting from the initial time $t=0$. Consider the $k$--th iteration in the feedback loop.
The sensor requests access to the network and after a waiting time $\Delta_{2k}^{\req}\in [0, \Delta_{\max}^{\req}]$, it sends at time $t_{2k}$ the latest available sample $y_k=[y(t_{2k})]_{\mu_{x}}$ where $\mu_{x}$ is the precision of the quantizer that follows the sensor in the NCS scheme in Figure \ref{NCSpic}.
The \emph{sensor-to-controller (sc)} link of the network introduces a delay $\Delta_{2k}= \Delta_{\send}^{\scc} + \Delta_{2k}^{\delay}$, with $\Delta_{2k}^{\delay}\in [\Delta^{\delay}_{\min},\Delta_{\max}^{\delay}]$, where $\Delta_{\send}^{\scc} = \lceil\log_{2} \vert [X]_{\mu_{x}} \vert \rceil /B_{\max}$ is the minimum time required to send the information over the sensor-to-controller branch, assuming a digital communication channel of bandwitdh $B_{\max}\in\mathbb{R}^{+}$ (expressed in bits per second (bps)). The maximum network delay  $\Delta_{\max}^{\delay}$ takes into account congestion, other accesses to the communication channel, any kind of scheduling protocol and a finite number of subsequent packet dropouts, which is assumed to be uniformly bounded. After that time, the sensor sample reaches the symbolic controller, that is expressed in terms of the function $C:[X]_{\mu_{x}} \rightarrow[U]_{\mu_{u}}$,
with $\mu_{x}\leq \hat{\mu}_X$ and $\mu_{u}\leq \hat{\mu}_U$ so that the domain and co--domain of $C$ are non--empty. After a time $\Delta^{\ctrl}_{k}\in [\Delta^{\ctrl}_{\min},\Delta^{\ctrl}_{\max}]$, the value $u_{k+1}=C(y_k)$ is returned and it is sent through the network at time $t_{2k+1}$ (after a bounded waiting time $\Delta_{2k+1}^{\req}\in [0, \Delta_{\max}^{\req}]$).
The \emph{controller-to-actuator (ca)} link of the network introduces a delay $\Delta_{2k+1}=\Delta_{\send}^{\ca} + \Delta_{2k+1}^{\delay}$, where $\Delta_{2k+1}^{\delay}\in [\Delta^{\delay}_{\min},\Delta_{\max}^{\delay}]$ and $\Delta_{\send}^{\ca} = \lceil\log_{2} \vert [U]_{\mu_{u}} \vert \rceil/B_{\max}$ is the minimum time required to send the information over the controller-to-actuator branch of the network. After that time, the sample reaches the ZoH and at time $t=A_{k+1} \tau$ the ZoH is refreshed to the control value $u_{k+1}$, with $ A_{k+1}=\lceil (t_{2k+1}+\Delta_{2k+1})/\tau\rceil$. The next iteration starts and the sensor requests access to the network again.
Consider now the sequence of control values $\{u_{k}\}_{k\in\mathbb{N}_0}$. Each value is held for $N_k=A_{k+1}-A_{k}$ sampling intervals. Due to the bounded delays, one gets $N_k \in [{N}_{\min}; {N}_{\max}]$, with:
\begin{equation}
{N}_{\min} =\left\lceil \Delta_{\min}/\tau \right\rceil, \quad {N}_{\max}  = \left\lceil \Delta_{\max}/\tau \right \rceil,
\label{eq:N_min-max}
\end{equation}
where we set $\Delta_{\min}=\Delta_{\send}^{\scc}+\Delta_{\min}^{\ctrl}+\Delta_{\send}^{\ca}+2\Delta_{\min}^{\delay}$, $\Delta_{\max}=\Delta_{\send}^{\scc}+\Delta_{\max}^{\ctrl}+\Delta_{\send}^{\ca}+2\Delta^{\req}_{\max}+2\Delta_{\max}^{\delay}$.
In the sequel we refer to the described NCS by $\Sigma$ and to a trajectory of $\Sigma$ with initial state $x_0$ and control input $u$ by $\mathbf{x}(.,x_{0},u)$.

\section{Systems, Approximate Equivalence and Composition}\label{subsec:ApproxEquiv}
We use the notion of system as a unified mathematical framework to describe NCS as well as their symbolic models.

\begin{definition}
\cite{paulo}
A system $S$ is a sextuple $S=(X,X_{0},U,\rTo,Y,H)$
consisting of:
\begin{itemize}
\item a set of states $X$;
\item a set of initial states $X_{0}\subseteq X$;
\item a set of inputs $U$;
\item a transition relation $\rTo \subseteq X\times U\times X$;
\item a set of outputs $Y$;
\item an output function $H:X\rightarrow Y$.
\end{itemize}
A transition $(x,u,x^{\prime})\in\rTo$ is denoted by $x\rTo^{u}x^{\prime}$. For such a transition, state $x^{\prime}$ is called a $u$-successor, or simply a successor, of state $x$.
\end{definition}

A state run of $S$ is a (possibly infinite) sequence of transitions $x_{0} \rTo^{u_{1}} x_{1} \rTo^{u_{2}} \dots$ with $x_0\in X_0$. An output run is a (possibly infinite) sequence $\{y_i\}_{i\in\mathbb{N}_0}$ such that there exists a state run with $y_i=H(x_i)$, $i\in\mathbb{N}_0$.
System $S$ is said to be:
\begin{itemize}
\item \textit{countable} if $X$ and $U$ are countable sets;
\item \textit{symbolic} if $X$ and $U$ are finite sets;
\item \textit{metric} if the output set $Y$ is equipped with a metric $d:Y\times Y\rightarrow\mathbb{R}_{0}^{+}$;
\item \textit{deterministic} if for any $x\in X$ and $u\in U$ there exists at most one state $x^{\prime}\in X$ such that $x \rTo^{u} x^{\prime}$ for some $u\in U$;
\item \textit{non--blocking} if for any $x\in X$ there exists at least one state $x^{\prime}\in X$ such that $x \rTo^{u} x^{\prime}$ for some $u\in U$;
\item \textit{accessible}, if for any $x\in X$ there exists a finite number of transitions $x_{0} \rTo^{u_{1}} x_{1} \rTo^{u_{2}} \ldots \rTo^{u_{N}} x$ from an initial state $x_{0}\in X_{0}$ to state $x$.
\end{itemize}
\begin{definition}
Given two systems $S_{i}=(X_{i},X_{0,i},U_{i},$ $\rTo_{i},Y_{i},H_{i})$ ($i=1,2$), $S_{1}$ is a \textit{sub--system} of $S_{2}$, denoted $S_{1} \sqsubseteq S_{2}$, if $X_{1}\subseteq X_{2}$, $X_{0,1}\subseteq X_{0,2}$, $U_{1}\subseteq U_{2}$, $\rTo_{1}\subseteq \rTo_{2}$, $Y_{1}\subseteq Y_{2}$, and $H_{1}(x)=H_{2}(x)$ for any $x\in X_{1}$.
\end{definition}

In the sequel we consider (alternating) approximate simulation relations \cite{paulo} to relate properties of NCS and symbolic models.

\begin{definition}
\cite{AB-TAC07,PolaSIAM2009}
\label{ASR}
Let $S_{i}=(X_{i},X_{0,i},U_{i},\rTo_{i},$ $Y_{i},H_{i})$ ($i=1,2$) be metric systems with the same output sets $Y_{1}=Y_{2}$ and metric $d$, and let $\varepsilon\in\mathbb{R}^{+}_{0}$ be a given precision. Consider a relation $\mathcal{R}\subseteq X_{1}\times X_{2}$ satisfying the following conditions:
\begin{itemize}
\item[(i)] $\forall x_{1}\in X_{0,1}$ $\exists x_{2}\in X_{0,2}$ such that $(x_{1},x_{2})\in \mathcal{R}$;
\item[(ii)] $\forall (x_{1},x_{2})\in \mathcal{R}$, \mbox{$d(H_{1}(x_{1}),H_{2}(x_{2}))\leq\varepsilon$}.
\end{itemize}

Relation $\mathcal{R}$ is an \emph{$\varepsilon$--approximate simulation relation} from $S_{1}$ to $S_{2}$ if it enjoys conditions (i), (ii) and the following one:
\begin{itemize}
\item[(iii)] $\forall (x_{1},x_{2})\in \mathcal{R}$ \mbox{$\forall x_{1}\rTo_{1}^{u_{1}}x'_{1}$} \mbox{$\exists x_{2}\rTo_{2}^{u_{2}}x'_{2}$} such that $(x^{\prime}_{1},x^{\prime}_{2})\in \mathcal{R}$.
\end{itemize}
System $S_{1}$ is $\varepsilon$--simulated by $S_{2}$ or $S_{2}$ $\varepsilon$--simulates $S_{1}$, denoted \mbox{$S_{1}\preceq_{\varepsilon}S_{2}$}, if there exists an $\varepsilon$--approximate simulation relation from $S_{1}$ to $S_{2}$.
Relation $\mathcal{R}$ is an \emph{alternating $\varepsilon$--approximate ($A\varepsilon A$) simulation relation} from $S_{1}$ to $S_{2}$ if it enjoys conditions (i), (ii) and the following one:
\begin{itemize}
\item[(iii$'$)] $\forall (x_{1},x_{2})\in\mathcal{R}$ $\forall u_{1}\in U_{1}$ $\exists u_{2}\in U_{2}$ \mbox{$\forall x_{2}\rTo_{2}^{u_{2}}x'_{2}$} \mbox{$\exists x_{1}\rTo_{1}^{u_{1}}x'_{1}$} such that $(x_{1}^{\prime},x_{2}^{\prime} )\in\mathcal{R}$.
\end{itemize}
System $S_{1}$ is alternating $\varepsilon$--simulated by $S_{2}$ or $S_{2}$ alternating $\varepsilon$--simulates $S_{1}$, denoted \mbox{$S_{1}\preceq_{\varepsilon}^{\alt} S_{2}$}, if there exists an $A\varepsilon A$ simulation relation from $S_{1}$ to $S_{2}$.
\end{definition}
For more details on the above notions we refer to \cite{paulo,AB-TAC07,PolaSIAM2009}.
We conclude this section with the notion of approximate feedback composition, that is employed in the sequel to capture feedback interaction between non-deterministic systems and symbolic controllers.

\begin{definition}
\cite{paulo}
\label{composition}
Consider a pair of metric systems \mbox{$S_{i}=(X_{i},X_{0,i},U_{i},\rTo_{i},Y_{i},H_{i})$} ($i=1,2$) with the same output sets $Y_{1}=Y_{2}$ and metric $d$. Let $\mathcal{R}$ be an $A \theta A$ simulation relation from $S_{2}$ to $S_{1}$. The $\theta$--approximate feedback composition of $S_{1}$ and $S_{2}$, with composition relation $\mathcal{R}$, is the system $S_{1}\times^{\mathcal{R}}_{\theta}S_{2}=(X,X_{0},U,\rTo,Y,H)$, where
\begin{itemize}
\item $X=\mathcal{R}^{-1}$;
\item $X_{0}=X\cap(X_{0,1}\times X_{0,2})$;
\item $U=U_{1}$;
\item $(x_{1},x_{2})\rTo^{u_1}(x_{1}^{\prime},x_{2}^{\prime})$ if $x_{1}\rTo_{1}^{u_1}x_{1}^{\prime}$ and  $x_{2}\rTo_{2}^{u_2}x_{2}^{\prime}$; 
\item $Y=Y_{1}$;
\item $H(x_{1},x_{2})=H_{1}(x_{1})$ for any $(x_{1},x_{2})\in X$.
\end{itemize}
\end{definition}


\section{Symbolic Models for NCS}\label{sec:SymbolicModels}

In this section we propose symbolic models that approximate NCS in the sense of (alternating) approximate simulation.
For notational simplicity we denote by $u$ any constant control input $\tilde{u}\in\mathcal{U}$ s.t. $\tilde{u}(t)=u$ at all times $t\in\mathbb{R}^{+}$. Set $X_e=\cup_{N\in [N_{\min};N_{\max}]} X^{N}$.

\begin{definition}
\cite{BorriHSCC12}
Given the NCS $\Sigma$, consider the system
\[
S(\Sigma)=(X_{\tau},X_{0,\tau},U_{\tau},\rTo_{\tau},Y_{\tau},H_{\tau})
\]
where:

\begin{itemize}
\item $X_{\tau}$ is the subset of $ X_0 \cup X_e$ such that for any $x=(x_{1},x_{2},...,x_{N})\in X_{\tau}$, with $N\in [N_{\min};N_{\max}]$, the following conditions hold:
    \begin{align}
    & x_{i+1}  =\mathbf{x}(\tau,x_{i},u^{-}), \quad i\in [1;N-2]); \label{eq:states_1}\\
    & x_{N} =\mathbf{x}(\tau,x_{N-1},u^{+}); \label{eq:states_2}
    \end{align}
for some constant functions $u^{-}$, $u^{+}\in [U]_{\mu_{u}}$;
\item $X_{0,\tau}=X_0$;
\item $U_{\tau}=[U]_{\mu_{u}}$;
\item $x^{1}\rTo^{u}_{\tau} x^{2}$, where
\[
\left\{
\begin{array}
{l}
x_{i+1}^{1} = \mathbf{x}(\tau,x_{i}^{1},u^{-}_{1}), i\in [1;N_{1}-2];\\
x_{N_{1}}^{1} = \mathbf{x}(\tau,x^1_{N_{1}-1},u^{+}_{1});
\\
x_{i+1}^{2} = \mathbf{x}(\tau,x_{i}^{2},u^{-}_{2}), i\in [1;N_{2}-2]; \\
x_{N_{2}}^{2} = \mathbf{x}(\tau,x^2_{N_{2}-1},u_2^{+}) ;
\\
u_2^{-} = u_1^{+}; \\
u_2^{+} = u; \\
x^{2}_{1} = \mathbf{x}(\tau,x^1_{N_{1}},u^{-}_{2}) ;
\end{array}
\right.
\]
for some $N_{1},N_{2}\in [N_{\min};N_{\max}]$;
\item $Y_{\tau}=X_{\tau}$;
\item $H_{\tau}(x)=x$.
\end{itemize}

\end{definition}

Note that $S(\Sigma)$ is non-deterministic because, depending on the values of $N_{2}$, more than one $u$--successor of $x^{1}$ may exist.
Since the state vectors of $S(\Sigma)$ are built from trajectories of $\Sigma$ sampled every $\tau$ time units, $S(\Sigma)$ collects all the information of the NCS $\Sigma$ available at the sensor (see Figure \ref{NCSpic}) as formally stated in Theorem 5.1 of \cite{BorriHSCC12}.
System $S(\Sigma)$ can be regarded as metric with the metric $d_{Y_{\tau}}$ on $Y_{\tau}$ naturally induced by the metric $d_{X}(x_{1},x_{2})=\Vert x_{1} - x_{2}\Vert$ on $X$, as follows. Given any $x^{i}=(x_{1}^{i},x_{2}^{i},...,x_{N_{i}}^{i})$, $i=1,2$, we set $d_{Y_{\tau}}(x^{1},x^{2})=
\max_{i\in[1;N]}\Vert x^{1}_{i}-x^{2}_{i}\Vert$, if $N_{1}= N_{2}=N$ and $d_{Y_{\tau}}(x^{1},x^{2})=+\infty$, otherwise.
Although system $S(\Sigma)$ contains all the information of the NCS $\Sigma$ available at the sensor, it is not a finite model. In the following, we propose a system that approximates $S(\Sigma)$ and is symbolic. A key property for our developments is the notion of incremental forward completeness, as recalled hereafter.

\begin{definition}
\label{dFC}
\cite{MajidTAC11}
Control system (\ref{NCSeq1}) is incrementally forward complete ($\delta$-FC) if it is forward complete and there exists a continuous function \mbox{$\beta: \mathbb{R}_0^+\times\mathbb{R}_0^+\rightarrow\mathbb{R}_0^+$} such that for every $s\in\mathbb{R}^+$, the function $\beta(\cdot,s)$ belongs to class $\mathcal{K}_{\infty}$, and for any $x_{1},x_{2}\in X$, any $\tau\in\mathbb{R}^+$, and any $u\in\mathcal{U}$, the following condition is satisfied for all $t\in[0,\tau]$:
\[
\left\Vert \mathbf{x}(t,x_1,u)-\mathbf{x}(t,x_2,u)\right\Vert\leq\beta(\Vert{x_1}-x_2\Vert,t). %
\]
\end{definition}

Incremental forward completeness requires the distance between two arbitrary trajectories to be bounded by a continuous function capturing the mismatch between  initial conditions. The class of $\delta$-FC control systems is rather large and includes also some subclasses of unstable control systems; for instance unstable linear systems are $\delta$-FC.
The notion of $\delta$-FC can be described in terms of Lyapunov-like functions.
\begin{definition}
\label{delta_FC_Lyapunov}
A smooth function $V:X\times X\rightarrow\mathbb{R}$ is called a $\delta$--FC Lyapunov function for the control system (\ref{NCSeq1})  if there exist $\lambda\in\mathbb{R}$ and $\mathcal{K}_{\infty}$ functions $\underline{\alpha}$ and $\overline{\alpha}$ such that, for any $x_{1},x_{2}\in X$ and any $u\in U$, the following conditions hold true:
\begin{itemize}
\item[(i)] $\underline{\alpha}(\Vert{x_{1}-x_{2}}\Vert)\leq V(x_{1},x_{2})\leq\overline{\alpha}(\Vert{x_{1}-x_{2}}\Vert)$,
\item[(ii)] $\frac{\partial{V}}{\partial{x_{1}}} f(x_{1},u)+\frac{\partial{V}}{\partial{x_{2}}} f(x_{2},u) \leq \lambda V(x_{1},x_{2})$.
\end{itemize}
\end{definition}
The existence of a $\delta$-FC Lyapunov function was proven in \cite{MajidTAC11} to be a sufficient condition for $\delta$-FC of a control system.
In the following we suppose that the control system $P$ in the NCS $\Sigma$ enjoys the following properties:
\begin{itemize}
\item[(H1)] There exists a $\delta$--FC Lyapunov function $V$ satisfying the inequality (ii) in Definition \ref{delta_FC_Lyapunov} for some $\lambda\in\mathbb{R}$;
\item[(H2)] There exists a $\mathcal{K}_{\infty}$ function $\gamma$ such that 
$V(x,x^{\prime})-V(x,x^{\prime \prime})\leq\gamma(\Vert{x^{\prime}-x^{\prime \prime}}\Vert)$, for every $x,x^{\prime},x^{\prime\prime}\in X$.
\end{itemize}

Given a design parameter $\eta \in \mathbb{R}^{+}$, define the following system
\[
S_{\ast}(\Sigma)=(X_{\ast},X_{0,\ast},U_{\ast},\rTo_{\ast},Y_{\ast},H_{\ast})
\]
where:

\begin{itemize}
\item $X_{\ast}$ is the subset of $[X_0 \cup X_e]_{\mu_{x}}$ such that for any $x^{\ast}=(x^{\ast}_{1},x^{\ast}_{2},...,x^{\ast}_{N})\in X_{\ast}$ with $N\in [N_{\min};N_{\max}]$ the following condition holds:
\begin{align}
& V(\mathbf{x}(\tau,x^{\ast}_{i},u^{-}_{\ast}),x^{\ast}_{i+1}) \leq e^{\lambda \tau} \underline{\alpha}(\eta)+\gamma(\mu_{x}),  \quad i\in [1;N-2]; \label{eq:symb_states_1} \\
& V(\mathbf{x}(\tau,x^{\ast}_{N-1},u^{+}_{\ast}),x^{\ast}_{N}) \leq e^{\lambda \tau} \underline{\alpha}(\eta)+\gamma(\mu_{x}); \label{eq:symb_states_2}
\end{align}
for some constant functions $u^{-}_{\ast}$, $u^{+}_{\ast}\in [U]_{\mu_{u}}$;

\item $X_{0,\ast}=[X_0]_{\mu_{x}}$;

\item $U_{\ast}=[U]_{\mu_{u}}$;

\item $x^{1}\rTo^{u_{\ast}}_{\ast} x^{2}$, where
\[
\left\{
\begin{array}
{l}
V(\mathbf{x}(\tau,x_{i}^{1},u^{-}_{1}),x_{i+1}^{1}) \leq e^{\lambda \tau} \underline{\alpha}(\eta)+\gamma(\mu_{x}),
\quad \forall i\in [1;N_{1}-2];\\
V(\mathbf{x}(\tau,x^1_{N_{1}-1},u^{+}_{1}),x_{N_{1}}^{1}) \leq e^{\lambda \tau} \underline{\alpha}(\eta)+\gamma(\mu_{x}); \\
V(\mathbf{x}(\tau,x_{i}^{2},u^{-}_{2}),x_{i+1}^{2})   \leq e^{\lambda \tau} \underline{\alpha}(\eta)+\gamma(\mu_{x}),
\quad \forall i\in [1;N_{2}-2];\\
V(\mathbf{x}(\tau,x^2_{N_{2}-1},u^{+}_{2}),x_{N_{2}}^{2}) \leq e^{\lambda \tau} \underline{\alpha}(\eta)+\gamma(\mu_{x}); \\
u_2^{-} =  u_1^{+}; \\
u_2^{+} =  u_{\ast}; \\
V(\mathbf{x}(\tau,x^1_{N_{1}},u_1^{+}),x^{2}_{1})\leq e^{\lambda \tau} \underline{\alpha}(\eta)+\gamma(\mu_{x});
\end{array}
\right.
\]
for some $N_{1},N_{2}\in [N_{\min};N_{\max}]$;

\item $Y_{\ast}=X_{\tau}$;
\item $H_{\ast}(x^{\ast})=x^{\ast}$.
\end{itemize}

System $S_{\ast}(\Sigma)$ is metric when we regard the set of outputs $Y_{\ast}$ as being equipped with the metric $d_{Y_{\tau}}$.
We now have all the ingredients to present one of the main results of this paper.

\begin{theorem}
\label{thmain}
Consider the NCS $\Sigma$ and suppose that the control system $P$ enjoys properties (H1) and (H2). Then for any desired precision $\varepsilon\in\mathbb{R}^{+}$, any sampling time $\tau\in\mathbb{R}^{+}$, any state quantization $\mu_{x}\in\mathbb{R}^{+}$ and any choice of the design parameter $\eta\in\mathbb{R}^{+}$ satisfying the inequality
\begin{equation}
\mu_{x} \leq \min \{ \hat{\mu}_X,\overline{\alpha}^{-1}(\underline{\alpha}(\varepsilon)) \} \leq \eta,
\label{bisim_condition1}
\end{equation}
we have $S_{\ast}(\Sigma) \preceq_{\varepsilon}^{\alt} S(\Sigma) \preceq_{\varepsilon} S_{\ast}(\Sigma)$.
\label{polaut}
\end{theorem}

\begin{proof}
First we prove that $S_{\ast}(\Sigma) \preceq_{\varepsilon}^{\alt} S(\Sigma)$, according to Definition \ref{ASR}. Consider the relation $\mathcal{R}\subseteq X_{\ast}\times X_{\tau}$ defined by $(x^{\ast},x)\in\mathcal{R}$ if and only if:
\begin{itemize}
\item $x^{\ast}=(x^{\ast}_{1},x^{\ast}_{2},...,x^{\ast}_{N})$, $x=(x_{1},x_{2},...,x_{N})$, for some $N\in [N_{\min};N_{\max}]$;
\item $V(x_{i}^{\ast},x_{i})\leq \underline{\alpha}(\varepsilon)$ for $i\in [1;N]$;
\item Eqns. (\ref{eq:states_1}), (\ref{eq:states_2}), (\ref{eq:symb_states_1}), (\ref{eq:symb_states_2}) hold for some $u^{-}=u^{-}_{\ast}$ and $u^{+}=u^{+}_{\ast}$.
\end{itemize}

Conditions (i) and (ii) in Definition \ref{ASR} can be proven by using similar arguments employed in the proof of Theorem 5.8 in \cite{BorriHSCC12}.
We now show that condition (iii$'$) in Definition \ref{ASR} holds. Consider any $(x^{\ast},x)\in\mathcal{R}$, with $x^{\ast}=(x^{\ast}_{1},x^{\ast}_{2},...,x^{\ast}_{N})$, $x=(x_{1},x_{2},...,x_{N})$, for some $N\in [N_{\min};N_{\max}]$, and any $u_{\ast}\in U_{\ast}$; then pick $u=u_{\ast}\in U_{\tau}$ and consider any transition $x \rTo^u_{\tau} \bar{x}$, with $\bar{x}=(\bar{x}_1,\bar{x}_2,...,\bar{x}_{\bar{N}})$, for some $\bar{N}\in [N_{\min};N_{\max}]$. Pick $\bar{x}^{\ast}=(\bar{x}^{\ast}_1,\bar{x}^{\ast}_2,...,\bar{x}^{\ast}_{\bar{N}})$ defined by $\bar{x}^{\ast}_i=[\bar{x}_i]_{\mu_x}$ for all $i$.
We now prove that $x^{\ast} \rTo^{u_{\ast}} \bar{x}^{\ast}$ is a transition of $S_{\ast}(\Sigma)$. First, from condition (i) in Definition \ref{delta_FC_Lyapunov}, the definition of $\bar{x}$ and the first inequality in (\ref{bisim_condition1}), one can write:
\begin{equation}
V(\bar{x}^{\ast}_{i},\bar{x}_{i}) \leq \overline{\alpha}(\mu_x) \leq \overline{\alpha}(\overline{\alpha}^{-1}(\underline{\alpha}(\varepsilon))) = \underline{\alpha}(\varepsilon)
\label{eq:bound_epsilon}
\end{equation}
for all $i$. By Assumption (H1), condition (ii) in Definition \ref{delta_FC_Lyapunov} writes:
\begin{equation}
\frac{\partial{V}}{\partial{x_{N}^{\ast}}} f(x_{N}^{\ast},u_{\ast}^+) + \frac{\partial{V}}{\partial{x_{N}}} f(x_{N},u^+) \leq \lambda V(x_{N}^{\ast},x_{N}). \\
\label{eq:ineq}
\end{equation}

By considering Assumption (H2), the definitions of $\mathcal{R}$ and $ S(\Sigma)$, and by integrating the previous inequality, the following holds:
\begin{align}
\label{eq_chain_of_ineq}
V(\mathbf{x}(\tau,x^{\ast}_{N},u_{\ast}^+),\bar{x}^{\ast}_{1})   & \leq V(\mathbf{x}(\tau,x^{\ast}_{N},u_{\ast}^+),\bar{x}_1) +\gamma(\Vert{\bar{x}_1-\bar{x}^{\ast}_{1}}\Vert) \\
&  \leq  e^{\lambda \tau} V(x^{\ast}_N,x_N)+\gamma(\Vert{\bar{x}_1-\bar{x}^{\ast}_{1}}\Vert) \nonumber \\
& \leq  e^{\lambda \tau} \underline{\alpha}(\varepsilon)+\gamma(\mu_{x}) \leq  e^{\lambda \tau} \underline{\alpha}(\eta)+\gamma(\mu_{x}),\nonumber
\end{align}
where condition $\varepsilon \leq \eta$ in (\ref{bisim_condition1}) has been used in the last step. By similar computations, it is possible to prove that the inequality in (\ref{eq:bound_epsilon}) implies:
\begin{align}
& V(\mathbf{x}(\tau,\bar{x}^{\ast}_{i},u_{\ast}^+),\bar{x}^{\ast}_{i+1}) \leq e^{\lambda \tau} \underline{\alpha}(\eta)+\gamma(\mu_{x}), \text{ } i\in [1;\bar{N}-2];\label{eq:recursive}\\
& V(\mathbf{x}(\tau,\bar{x}^{\ast}_{N-1},u_{\ast}),\bar{x}^{\ast}_{N})
\leq  e^{\lambda \tau} \underline{\alpha}(\eta)+\gamma(\mu_{x}). \label{eq:recursive_2}
\end{align}

Hence, from the inequalities in (\ref{eq_chain_of_ineq})--(\ref{eq:recursive_2}) and from the definition of the transition relation in $S_{\ast}(\Sigma)$, the transition $x^{\ast} \rTo^{u_{\ast}} \bar{x}^{\ast}$ is in $S_{\ast}(\Sigma)$, implying with (\ref{eq:bound_epsilon}) that $(\bar{x}^{\ast},\bar{x})\in\mathcal{R}$, which concludes the proof of condition (iii$'$) of Definition \ref{ASR}.
We now prove $S(\Sigma) \preceq_{\varepsilon} S_{\ast}(\Sigma)$, according to Definition \ref{ASR}, by considering the relation  $\mathcal{R}^{-1}$.
We prove condition (iii) in Definition \ref{ASR}, because the proof of condition (i) is given in \cite{BorriHSCC12}, while condition (ii) is fulfilled for the relation $\mathcal{R}^{-1}$ because it has been proved to hold for $\mathcal{R}$.
Consider any $(x,x^{\ast})\in\mathcal{R}^{-1}$, with $x=(x_{1},x_{2},...,x_{N})$, $x^{\ast}=(x^{\ast}_{1},x^{\ast}_{2},...,x^{\ast}_{N})$, for some $N\in [N_{\min};N_{\max}]$, and any transition $x \rTo^{u} \bar{x}$ in $S(\Sigma)$, for some $u\in U_{\tau}$, with $\bar{x}=(\bar{x}_1,\bar{x}_2,...,\bar{x}_{\bar{N}})$ for some $\bar{N}\in [N_{\min};N_{\max}]$. Pick $\bar{x}^{\ast}=(\bar{x}^{\ast}_1,\bar{x}^{\ast}_2,...,\bar{x}^{\ast}_{\bar{N}})$ defined by $\bar{x}^{\ast}_i=[\bar{x}_i]_{\mu_x}$ for all $i$. By using similar arguments as in the proof of condition (iii$'$) of Definition \ref{ASR} for the relation $\mathcal{R}$, it is possible to show that the transition $x^{\ast} \rTo^{u_{\ast}} \bar{x}^{\ast}$, with $u_{\ast}=u$, is in $S_{\ast}(\Sigma)$, and that $V(\bar{x}_{i},\bar{x}^{\ast}_{i}) \leq \underline{\alpha}(\varepsilon)$ for all $i$, hence $(\bar{x},\bar{x}^{\ast})\in\mathcal{R}^{-1}$, which concludes the proof.
\end{proof}

This result is important because it provides symbolic models for \textit{possibly unstable} nonlinear NCS, with guaranteed approximation bounds. This result generalizes the ones in \cite{BorriHSCC12}, which instead require incrementally stable NCS.

\section{Robust symbolic Control Design}\label{sec:control}

We consider a control design problem where the NCS $\Sigma$ has to satisfy a given specification robustly with respect to the non-idealities of the communication network. Our specification is a collection of transitions $\rTo_{\bar{q}}\subseteq \bar{X}_{q}\times \bar{X}_{q}$, where $\bar{X}_{q}$ is a finite subset of $\mathbb{R}^{n}$. Given a set of initial states $\bar{X}_{q}^{0}\subseteq \bar{X}_{q}$, we now reformulate the specification in the form of the system
\[
\mathcal{Q}=(X_q,X^{0}_q,U_q,\rTo_{q},Y_q,H_q),
\]

where:

\begin{itemize}
\item $X_q$ is the subset of $ \bar{X}_{q}^0 \cup \left( \bigcup_{N\in [N_{\min};N_{\max}]} \bar{X}_q^{N} \right)$ such that for any $x=(x_{1},x_{2},...,x_{N})\in X_q$, with $N\in [N_{\min};N_{\max}]$, for any $i\in [1;N-1]$, the transition $x_{i} \rTo_q x_{i+1}$ is in $\rTo_{\bar{q}}$;

\item $X^{0}_q=\bar{X}_{q}^0$;

\item $U_q=\{ \bar{u}_q \}$, where $\bar{u}_q$ is a \emph{dummy} symbol;

\item $x^{1}\rTo_{q}^{\bar{u}_q} x^{2}$, where
$x^{1} =(x^{1}_{1},x^{1}_{2},...,x^{1}_{N_{1}})$, $x^{2} =(x^{2}_{1},x^{2}_{2},...,x^{2}_{N_{1}})$, $N_1,N_2 \in [N_{\min};N_{\max}]$ and the transition $x^1_{N_1} \rTo_q x^2_{1}$ is in $\rTo_{\bar{q}}$;

\item $Y_q=X_q$;

\item $H_q=1_{X_q}$,

\end{itemize}

where $N_{\min}$ and $N_{\max}$ are as in (\ref{eq:N_min-max}).
We are now ready to state the control problem that we address in this section.

\begin{problem}
\label{problem}
Consider the NCS $\Sigma$, a specification $\mathcal{Q}$ and a desired precision $\varepsilon\in\mathbb{R}^{+}$. Find a symbolic controller $C$, a parameter $\theta\in\mathbb{R}^{+}$ and a $A\theta A$ simulation relation $\mathcal{R}$ from $C$ to $S(\Sigma)$ such that:
\begin{itemize}
\item [(1)] $\varnothing \neq S(\Sigma)\times_{\theta}^\mathcal{R} C  \preceq_{\varepsilon} \mathcal{Q}$;
\item [(2)] $S(\Sigma)\times_{\theta}^\mathcal{R} C$ is non-blocking.
\end{itemize}
\end{problem}
Note that the approximate similarity inclusion in (1) requires the state trajectories of the NCS to be close to the ones of specification $\mathcal{Q}$ up to the accuracy $\varepsilon$ robustly with respect to the non-determinism imposed by the network. The non-blocking condition (2) prevents deadlocks in the interaction between the plant and the controller.
In the following definition, we provide the controller $C^{\ast}$ that is shown in the sequel to solve Problem \ref{problem}.

\begin{definition}
Let $C^{\ast}$ be the maximal non-blocking sub-system\footnote{Here maximality is defined with respect to the preorder induced by the notion of $A 0 A$ simulation.} $C$ of $S_{\ast}(\Sigma)$ such that $C \preceq_{\mu_x} \mathcal{Q}$ and $C \preceq^{\alt}_{0}  S_{\ast}(\Sigma)$.
\label{canon_contr}
\end{definition}

From the above definition it is easy to see that $C^{\ast}$ is symbolic.
The following technical result will be useful in the sequel.

\begin{lemma}
Let \mbox{$S_{i}=(X_{i},X_{0,i},U_{i},\rTo_{i},Y_{i},H_{i})$} ($i=1$, $2$, $3$) be metric systems with the same output sets $Y_{1}=Y_{2}=Y_{3}$ and metric $d$. Then the following statements hold:
\begin{itemize}
\item[(i)] for any $\varepsilon_{1}\leq \varepsilon_{2}$, $S_{1}\preceq^{\alt}_{\varepsilon_{1}}S_{2}$ implies \mbox{$S_{1}\preceq^{\alt}_{\varepsilon_{2}}S_{2}$};
\item[(ii)] if $S_{1}\preceq^{\alt}_{\varepsilon_{12}}S_{2}$ and $S_{2}\preceq^{\alt}_{\varepsilon_{23}}S_{3}$ then $S_{1}\preceq^{\alt}_{\varepsilon_{12}+\varepsilon_{23}}S_{3}$;
\item[(iii)] for any $\theta\in\mathbb{R}^{+}_{0}$ and any $A\theta A$ simulation relation $\mathcal{R}$ from $S_2$ to $S_1$, \\ \mbox{$S_{1} \times^{\mathcal{R}}_{\theta} S_{2} \preceq_{\theta} S_{2}$}.
\end{itemize}
\label{lemma1}
\end{lemma}
We are now ready to solve Problem \ref{problem}.

\begin{theorem}
Consider the NCS $\Sigma$ and the specification $\mathcal{Q}$. Suppose that the control system $P$ in $\Sigma$ enjoys Assumptions (H1) and (H2). Then for any desired precision $\varepsilon\in\mathbb{R}^{+}$, choose the parameters $\theta,\mu_{x},\eta\in\mathbb{R}^{+}$ such that:
\begin{align}
& \mu_{x}+\theta  \leq\varepsilon,\label{solving_condition_1}\\
& \mu_{x} \leq \min \{ \hat{\mu}_X,\overline{\alpha}^{-1}(\underline{\alpha}(\theta)) \} \leq \theta \leq \eta\text{.}\label{solving_condition_2}
\end{align}
Let $\bar{\mathcal{R}}$ be the maximal $A \theta A$ simulation relation\footnote{The maximal $A \theta A$ simulation relation is the unique $A \theta A$ simulation relation that contains all the $A \theta A$ simulation relations.} from $C^{\ast}$ to $S(\Sigma)$. If $\bar{\mathcal{R}}\neq \varnothing$, Problem \ref{problem} is solved with $C=C^{\ast}$ and $\mathcal{R}=\bar{\mathcal{R}}$.
\label{Mmain}
\end{theorem}

\begin{proof}
First we prove condition (1) of Problem \ref{problem}. From Definition \ref{canon_contr}, $C^{\ast} \preceq^{alt}_{0} S_{\ast}(\Sigma)$. Furthermore, condition (\ref{solving_condition_2}) implies that $S_{\ast}(\Sigma) \preceq_{\theta}^{alt} S(\Sigma)$ from Theorem \ref{polaut}. Hence from Lemma \ref{lemma1} (ii), $C^{\ast}\preceq_{\theta}^{alt} S(\Sigma)$. Let $\bar{\mathcal{R}}$ be the maximal $A \theta A$ simulation relation from $C^{\ast}$ to $S(\Sigma)$. From Lemma \ref{lemma1} (iii), $S(\Sigma)\times^{\bar{\mathcal{R}}}_{\theta} C^{\ast} \preceq_{\theta} C^{\ast}$. Since $C^{\ast} \preceq_{\mu_x} Q$ from Definition \ref{canon_contr}, by Proposition 2 in \cite{AB-TAC07} the above approximate similarity inclusions imply $S(\Sigma)\times^{\bar{\mathcal{R}}}_{\theta} C^{\ast} \preceq_{\varepsilon} Q$, which concludes the proof of condition (1) of Problem \ref{problem}.

We now show that condition (2) holds. Consider any state $(x,x_c)$ of $S(\Sigma)\times^{\bar{\mathcal{R}}}_{\theta}  C^{\ast}$. Pick any $u_{c}\in U_{c}(x_c)\neq \varnothing$ because $C^{\ast}$ is non-blocking. Since $(x_c,x)$ belongs to the maximal $A \theta A$ simulation $\bar{\mathcal{R}}$ relation from $C^{\ast}$ to $S(\Sigma)$, there exists $u\in U_{\tau}(x)$ s.t. for any $x\rTo_{\tau}^{u} x^{\prime}$ in $S(\Sigma)$ there exists $x_c\rTo_{c}^{u_c}x^{\prime}_{c}$ in $C^{\ast}$ with $(x^{\prime}_{c},x^{\prime})\in\bar{\mathcal{R}}$. Hence, from Definition \ref{composition}, the transition $(x,x_{c}) \rTo^u (x^{\prime},x^{\prime}_{c})$ is in $S(\Sigma)\times^{\bar{\mathcal{R}}}_{\theta} C^{\ast}$, implying that $S(\Sigma)\times^{\bar{\mathcal{R}}}_{\theta} C^{\ast}$ is non-blocking, which concludes the proof.
\end{proof}

\section{Integrated Design of Symbolic Controllers}\label{sec:integrated}

The construction of the symbolic controller $C^{\ast}$ relies upon the procedure illustrated in Algorithm \ref{alg3}.

\incmargin{1em}
\restylealgo{boxed}\linesnumbered
\begin{algorithm}[H]
\label{alg3}
\SetLine
\caption{Construction of the controller $C^{\ast}$.}
Compute the system $S_{\ast}(\Sigma)$\;
Compute the system $\mathcal{Q}$ from the transition relation $\rTo_{\bar{q}}$\;
Compute the controller $C^{\ast}$.
\end{algorithm}
\decmargin{1em}

This procedure is not efficient from the computational complexity point of view, because:
\begin{itemize}
\item[(i)] It requires the preliminary construction of the symbolic system $S_{\ast}(\Sigma)$, representing the NCS, and of the system $\mathcal{Q}$, representing the specification.
\item[(ii)] It considers the whole state space of the plant $P$, while a more efficient algorithm would consider only the accessible part\footnote{The \textit{accessible part} of a system $S$ is the unique accessible system $Ac(S)$ such that $S' \sqsubseteq Ac(S) \sqsubseteq S$, for any accessible system $S' \sqsubseteq S$. } of $P$. 
\end{itemize}

In order to cope with the drawbacks listed above, inspired by the integrated procedure developed in \cite{PolaTAC12} for the simpler case of symbolic control design of nonlinear systems, we now present a procedure that \textit{integrates each step of Algorithm \ref{alg3} in one algorithm}.
The pseudo-code of the proposed procedure is reported in Algorithm \ref{alg} and Algorithm \ref{alg2}. Algorithm \ref{alg} is the main one while Algorithm \ref{alg2} introduces function \textbf{BuildTree} that is used in Algorithm \ref{alg}. The outcome of Algorithm \ref{alg} is the symbolic controller $C^{\ast\ast}$. In the sequel, line $i$ of Algorithm $j$ will be recalled as line $j.i$.
Algorithm \ref{alg} proceeds as follows. In line 2.2  the set $\mathbf{X}_{target}$ of to-be-processed states is initialized and the set $Bad$ of blocking states is empty.
At each basic step, Algorithm \ref{alg} processes a (non--processed) state $x$ in line 2.4. The test in line 2.6 verifies the existence of a control input $u$ such that all the states (collected in the vector $\mathbf{x}(N_{\min}\tau{:}N_{\max}\tau,x,u)$) that are reachable from $x$ in the plant in time intervals from $N_{\min} \tau$ to $N_{\max} \tau$ are also reachable (up to the accuracy $\theta$) in the specification through a path of length between $N_{\min}$ and $N_{\max}$. If that happens, the control input $u$ is good for state $x$ (it is added to the controller in line 2.7) and function \textbf{BuildTree} is called (line 2.14) from all the states reached in the plant that are not equal to the state $x$ that is being processed (lines 2.11--2.12). If there exists a controller fulfilling the specification for all those states, the boolean variable $Found$ is set to $true$ and a solution is found (lines 2.24--2.25), otherwise it is guaranteed that $C^{\ast}$ defined in Definition \ref{canon_contr} is empty.
Algorithm \ref{alg2} (function \textbf{BuildTree}) checks the existence of a control input starting from the current state such that the specification is fulfilled robustly, up to the precision $\theta$. If that happens, the control input is added to the controller (line 3.5) and function \textbf{BuildTree} itself is called (line 3.13) recursively from all the states reached in the plant that have not been processed yet (lines 3.8--3.11). If there exists a controller fulfilling the specification for all those states, the function returns true (line 3.16), otherwise (line 3.19) it returns false and the current state is added to the set of bad states (line 3.20).
Termination, correctness and complexity of the integrated procedure are discussed in the remainder of this section. 

\incmargin{1em}
\restylealgo{boxed}\linesnumbered
\begin{algorithm}
\label{alg}
\SetLine
\caption{Integrated Symbolic Control Design.}
\textbf{Input:}
NCS $\Sigma$,
specification $\mathcal{Q}$,
precision $\varepsilon\in\mathbb{R}^{+}$,
quantization parameters $\theta,\mu_x,\eta \in \mathbb{R}^{+}$ satisfying the inequalities in (\ref{solving_condition_1}--\ref{solving_condition_2})\;
\textbf{Init:} $\mathbf{X}_{target}=\{x_p\in X_{0,\ast} : \exists x_{q}\in X^{0}_{q} :
\Vert x_{p} - x_{q} \Vert \leq \theta \}
$, $\mathbf{global}$ $Bad=\varnothing$, $\mathbf{global}$ $C^{\ast\ast}$, $found=false$\;
\While{$\mathbf{X}_{target}\neq \varnothing \wedge found==false$}
{
\textbf{choose} $x \in \mathbf{X}_{target}$\;
$\mathbf{U}_{target}=U_{\ast}$\;
\While{$\mathbf{U}_{target}\neq \varnothing \wedge found==false$}
{
\textbf{choose} $u \in \mathbf{U}_{target}$\;
 $C^{\ast\ast}=\varnothing$\;

\If{$\mathbf{x}(N_{\min}\tau{:}N_{\max}\tau,x,u) \text{ meets } \mathcal{Q} \text{
 up to } \theta$}
{
\For{$N=N_{\min}:N_{\max}$}
{
                			 \uIf{$\exists x_{c}\in \mathbf{Domain}(C^{\ast\ast}) : \Vert \mathbf{x}(N\tau,x,u)] - x_{c} \Vert \leq \theta$}
                			 {
                			 $Flag_N=true$;
                			 }
                			 \Else
                			 {
                			 ${Flag_N=\textbf{BuildTree}([\mathbf{x}(N\tau,x,u)]_{\mu_x})}$;
                       }
}

}
$found=\bigwedge_{N=N_{\min}}^{N_{\max}} Flag_N$\;
$\mathbf{U}_{target}=\mathbf{U}_{target}\setminus \{u\}$;
}
$\mathbf{X}_{target}=\mathbf{X}_{target}\setminus \{x\}$;
}

\uIf {$found==true$}{
$C^{\ast\ast}(x)=u$\;
\textbf{Controller found successfully!}
}
\Else {
$C^{\ast\ast}=\varnothing$\;
}
\textbf{output:} $C^{\ast\ast}$.
\end{algorithm}
\decmargin{1em}

\incmargin{1em}
\restylealgo{boxed}\linesnumbered
\begin{algorithm}
\label{alg2}
\SetLine
\caption{Recursive computation of subcontrollers.}
\textbf{Function }$flag=$\textbf{BuildTree}$(x)$\;
\textbf{Init:} $flag=false$, $\mathbf{U}_{target}=U_{\ast}$\;

\While{$\mathbf{U}_{target}\neq \varnothing \wedge flag==false$}
{
\textbf{choose} $u \in \mathbf{U}_{target}$\;

     			 $C^{\ast\ast}(x)=u$\;
     			 \If{$\mathbf{x}(N_{\min}\tau{:}N_{\max}\tau,x,u) \text{ meets } \mathcal{Q} \text{
 up } \theta$}
           {
     			 \For{$N=N_{\min}:N_{\max}$}
                			 {
                			 \uIf{$\exists x_{c}\in \mathbf{Domain}(C^{\ast\ast}) : \Vert \mathbf{x}(N\tau,x,u)] - x_{c} \Vert \leq \theta$}
                			 {
                			 $flag_N=true$;
                			 }
                			 \uElseIf{$[\mathbf{x}(N\tau,x,u)]_{\mu_x} \in Bad$}
                			 {
                			 $flag_N=false$;
                			 }
                			 \Else
                			 {
                			 ${flag_N=\textbf{BuildTree}([\mathbf{x}(N\tau,x,u)]_{\mu_x})}$;
                       }}

           {

           }$flag=\bigwedge_{N=N_{\min}}^{N_{\max}} flag_N$;
                			 }
     			      			}

           \If{$flag==false$}
     			      			{
     			      			$Bad=Bad\cup \{x\}$;
     			      			}
\end{algorithm}
\decmargin{1em}

\begin{theorem}
Algorithm \ref{alg} terminates in a finite number of steps.
\label{theo:termination}
\end{theorem}

\begin{proof}
Algorithm \ref{alg} terminates when there are no more states $x$ in $\mathbf{X}_{target}$ to be processed. Line 2.21 ensures that the iteration in line 2.3 is run at most once for any state $x$ in $\mathbf{X}_{target}$. Furthermore, the function \textbf{BuildTree} cannot be executed recursively on the same state (that would block the procedure). In fact, if condition in line 3.3 is satisfied, the execution of line 3.5 implies that state $x$ will enjoy the condition in line 3.8, hence preventing the recursive execution of line 3.13. Similarly, if a state $x$ becomes bad (line 3.20), it will satisfy condition in line 3.10 in successive iterations, hence preventing the recursive execution of line 3.13.
\end{proof}

We now show that the controller $C^{\ast\ast}$, synthesized in Algorithm \ref{alg}, solves Problem \ref{problem}.

\begin{theorem}
\label{theo:correctness}
Let $S_{cl}(\Sigma)$ be the maximal sub-system of $S(\Sigma)$ including all the transitions $x^{1}\rTo^{u}_{\tau} x^{2}$ in $S(\Sigma)$, with $x^{i}=(x_{1}^{i},x_{2}^{i},...,x_{N_{i}}^{i})$, $i=1,2$, such that $u=C^{\ast\ast}(x^1_{N_1})$. Then $S_{cl}(\Sigma) \preceq_{\varepsilon} \mathcal{Q}$ and  $S_{cl}(\Sigma)$ is non--blocking.
\end{theorem}

\begin{proof}
Condition (1) is ensured by the conditions in lines 2.9 and 3.6, that are required for adding control pairs $(x,u)$ to the controller. The non-blocking condition (2) is ensured because function \textbf{BuildTree} returns true only if all the states that are reached in a time between $N_{\min} \tau$ and $N_{\max} \tau$ are already in the domain of the controller (lines 3.8, 3.9 and 3.16). This implies that an execution from those states is well-defined and fulfills the specification.
\end{proof}
Theorem \ref{theo:correctness} extends the results reported in \cite{BorriHSCC12} from stable nonlinear control systems to $\delta$-FC nonlinear NCS.
%
Finally, a comparison of the following results shows that the space complexity of Algorithm \ref{alg} is smaller than or equal to the one of Algorithm \ref{alg3}.


\begin{proposition}
The space complexity of Algorithm \ref{alg3} is \mbox{$O(|[X]_{\mu_x}|^{N_{\max}-N_{\min}+1})$}.
\label{Prep:SpaceComplexityComparision1}
\end{proposition}

\begin{proof}
Algorithm \ref{alg3} requires the construction of the symbolic model $S_{\ast}(\Sigma)$ and the states of this model have $N_{\max}-N_{\min}+1$ components, implying a space complexity of \mbox{$O(|[X]_{\mu_x}|^{N_{\max}-N_{\min}+1})$}.
\end{proof}

\begin{proposition}
The space complexity of Algorithm \ref{alg} is \mbox{$O(|[X]_{\mu_x}|)$}.
\label{Prep:SpaceComplexityComparision2}
\end{proposition}

\begin{proof}
Algorithm \ref{alg} constructs a controller in form of a function $C:[X]_{\mu_{x}} \rightarrow[U]_{\mu_{u}}$ without requiring the construction of $S_{\ast}(\Sigma)$. Since the integrated controller keeps at most one input for each state, the complexity of that object is bounded by \mbox{$O(|[X]_{\mu_x}|)$}. The memory occupation of the set $Bad$ is also \mbox{$O(|[X]_{\mu_x}|)$}, while other variables have fixed sizes.
\end{proof}

%

\section{An Illustrative Example}\label{sec:example}

We consider the model of a unicycle $P$ described by the following differential equation:  
\begin{align}
\dot{x} &=\left[
\begin{array}
[l]{l}
\dot{x}_1\\
\dot{x}_2\\
\dot{x}_3
\end{array}
\right] =f(x,u)=
\left[
\begin{array}
[c]{c}
u_1 \cos (x_3)\\
u_1 \sin (x_3)\\
u_2
\end{array}
\right],\label{example1}
\end{align}

where the state $x$ belongs to the set $X=X_{0} =\left[-1,1\right[\times \left[-1,1\right[ \times \left[-\pi,\pi\right[$ and the control input $u$ belongs to the set $U =\left[-1,1\right[\times \left[-1,1\right[$. The state quantities are the $2$D-coordinates of the center of the vehicle and its orientation, while the inputs are the forward and angular velocity.
By choosing the quadratic Lyapunov-like function $V(x,x^{\prime})=0.5\, \Vert x-x^{\prime}\Vert_{2}^{2}$ it is possible to show that
control system (\ref{example1}) is $\delta$--FC.
The network/computation parameters are $B_{\max} =1 \text{ kbit}/s$, $\tau =0.2 s$, $\Delta_{\min}^{\ctrl}=0.001 s$, $\Delta_{\max}^{\ctrl} =0.01 s$, $\Delta^{\req}_{\max} =0.05 s$, $\Delta_{\min}^{\delay}=0.02 s$, $\Delta_{\max}^{\delay}=0.1 s$, resulting in $N_{\min}=1$, $N_{\max}=2$ from Eqn. (\ref{eq:N_min-max}).
In order to construct a symbolic model for $\Sigma$, we apply Theorem \ref{thmain}. Assumptions (H1)--(H2) are fulfilled for $P$ with $\lambda=2u_{1,\max}$ and $\gamma(r)=2\pi r$. For a precision $\varepsilon=0.15$, and the choice of parameters $\eta=0.11$, $\mu_x=0.02$ and $\mu_u=0.25$, the inequality in (\ref{bisim_condition1}) holds. 
We now consider a specification given in the form of a motion planning problem with respect to the position variables $x_1$ and $x_2$ of the unicycle. Starting from the origin, the vehicle is required to follow a trajectory visiting (in order) the $4$ regions of the plane $Z_1=[0,1[ \times [0,1[$, $Z_2=[-1,0[ \times [0,1[$, $
Z_3 =[-1,0[ \times [-1,0[$, and $Z_4 =[0,1[ \times [-1,0[$,
to finally go back to a neighbourhood of the origin.
For the choice of the interconnection parameter $\theta=0.9 \varepsilon$, Theorem \ref{Mmain} holds and the controller $C^{\ast}$ from Definition (\ref{canon_contr}) solves the control problem. We also solve the problem by means of the integrated procedure illustrated in Section \ref{sec:integrated} and in the following we compare the results in terms of the computational complexity needed to construct $C^{\ast}$ and $C^{\ast\ast}$.
The total memory occupation and time required to construct $C^{\ast\ast}$ are respectively $1345$ integers and $916\,$s. We did not compute the controller $C^{\ast}$; estimates of space complexity and time complexity in constructing $C^{\ast}$ result respectively in $5.8\cdot10^{12}$ integers and $4.19\cdot 10^6\,$s.
In Figures \ref{fig:sim1}--\ref{fig:sim2}, we show the simulation results for a particular realization of the network uncertainties: it is easy to see that the specifications are indeed met.

\begin{figure}[ht]
\begin{center}
\includegraphics[scale=0.7]{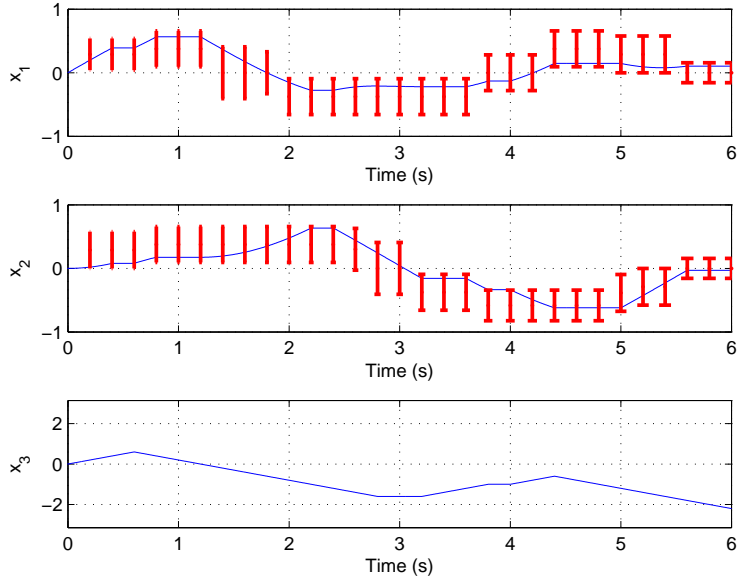}
\caption{State trajectory of the NCS $\Sigma$.}
\label{fig:sim1}
\end{center}
\end{figure}

\begin{figure}[ht]
\begin{center}
\includegraphics[scale=0.7]{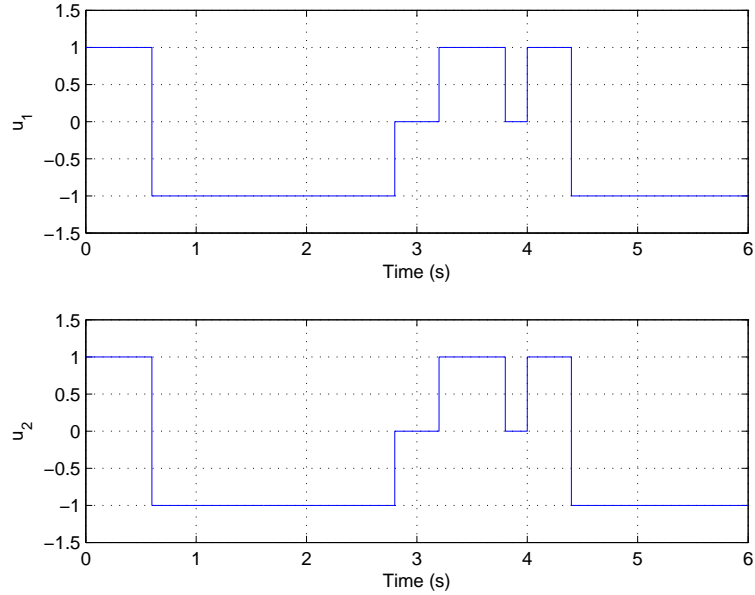}
\caption{Control input for the NCS $\Sigma$.}
\label{fig:sim2}
\end{center}
\end{figure}

\section{Conclusions}\label{sec:conclusion}
In this paper we proposed an integrated symbolic design approach to nonlinear NCS. Under the assumption of incremental forward completeness, symbolic models were derived which approximate NCS in the sense of (alternating) approximate simulation. Symbolic control design of NCS was then addressed where specifications are expressed in terms of automata. Finally efficient algorithms were proposed which integrate the construction of symbolic models with the design of robust symbolic controllers. 

\bibliographystyle{abbrv}
\bibliography{biblio1}

\end{document}